\documentclass[12pt,draftcls,a4paper,onecolumn]{IEEEtran}
\usepackage[english]{babel}
\usepackage{url}
\usepackage{graphicx}
\usepackage{amsmath}
\usepackage{amsfonts}
\usepackage{amssymb}

\ifx\pdftexversion\undefined
\else
\fi            

\newtheorem{proposition}{Proposition}

\def\B{\mathcal{B}}

\def\nontrans{\texttt{Non-transitive}}
\def\polish{\texttt{Transitive-predefined}}
\def\russianT{\texttt{Transitive-ab-initio}}
\def\russianNT{\texttt{Non-tree-transitive}}

\def\q#1{\mathtt{#1}}

\def\1{\q{1}}
\def\0{\q{0}}

\newcommand{\al}{\alpha}
\newcommand{\II}{{\cal{I}}}
\newcommand{\G}{\mathcal{B}}

\begin{document}

\title{On subset seeds for protein alignment\thanks{A preliminary
    version of this paper appeared in the proceedings of the
    ALBIO'08 workshop (Vienna, Austria
July 7-9, 2008)}
}

\author{Mikhail Roytberg, Anna Gambin, Laurent No{\'e}, S\l{}awomir
  Lasota, Eugenia Furletova, Ewa Szczurek, Gregory Kucherov 
\thanks{Mikhail Roytberg and Eugenia Furletova are with the 
Institute of Mathematical Problems in Biology,
Pushchino, Moscow Region, 142290, Russia, e-mail: {\tt
  mroytberg@mail.ru, furletova@impb.psn.ru}}
\thanks{Anna Gambin and S\l{}awomir Lasota are with the
Institute of Informatics, Warsaw University, Banacha 2, 02-097, Poland, \texttt{\{aniag|S.Lasota\}@mimuw.edu.pl}}
\thanks{Ewa Szczurek is with the 
Max Planck Institute for Molecular Genetics,
Computational Molecular Biology,
Ihnestr. 73,
14195 Berlin, Germany, \texttt{ewa.szczurek@molgen.mpg.de}}
\thanks{Gregory Kucherov and Laurent No\'e are with 
LIFL/CNRS/INRIA, B\^at. M3, Campus Scientifique,
59655 Villeneuve d'Ascq C\'edex, France, \texttt{\{Gregory.Kucherov|Laurent.Noe\}@lifl.fr}}
}



\maketitle
\begin{abstract}
We apply the concept of {\em subset seeds} proposed in
\cite{KucherovNoeRoytberg06} to similarity search in protein
sequences. The main question studied is the design of efficient 
{\em seed alphabets} to construct seeds with optimal
sensitivity/selectivity trade-offs. We propose several different
design methods and use them to construct several alphabets. We then
perform a comparative analysis of seeds built over those alphabets and
compare them with the standard {\sc Blastp} seeding method
\cite{BLAST90,GBLAST97}, as well as 
with the family of vector seeds proposed in \cite{BrownTCBB05}. While
the formalism of subset seeds is less expressive (but less costly to
implement) than the cumulative
principle used in {\sc Blastp} and vector seeds, our seeds show a
similar or even better performance than {\sc Blastp} on Bernoulli
models of proteins compatible with the common BLOSUM62
matrix. Finally, we perform a large-scale benchmarking of our seeds
against several main databases of protein alignments. Here again, the
results show a comparable or better performance of our seeds {\em vs.} 
{\sc Blastp}. 
\end{abstract}

\begin{keywords}
protein sequences, protein databases, local alignment, similarity search, seeds, subset seeds, multiple seeds, seed alphabet, sensitivity, selectivity
\end{keywords}

\section{Introduction}

Similarity search in protein sequences is probably the most classical
bioinformatics problem, and a commonly used algorithmic solution is
implemented in the ubiquitous {\sc Blast} software
\cite{BLAST90,GBLAST97}. On the other hand, similarity search
algorithms for nucleotide sequences (DNA, RNA) underwent
several years ago a significant improvement due to the idea of 
{\em spaced seeds} and its various generalizations
\cite{PatternHunter02,BLAT02,PatternHunter04,NoeKucherovNAR05,MakGelfandBensonBioinformatics06,CsurosMaAlgorithmica07}. 
This development, however, has little affected protein sequence
comparison, although improving the speed/precision trade-off for protein
search would be of great value for numerous bioinformatics projects. 
Due to a bigger alphabet
size, protein seeds are much shorter (typically 2-5 letters instead of
10-20 letters in the DNA case) and also letter identity is much less
relevant in defining hits than in the DNA case. For these reasons, the
spaced seeds technique might seem not to apply directly to protein
sequence comparison. 

Recall that {\sc Blast} applies quite different approaches to protein
and DNA sequences to define a hit. In the DNA case, a hit is defined as a
short pattern of identically matching nucleotides whereas in the
protein case, a hit is defined through a {\em cumulative}
contribution of a few amino acid matches (not necessarily identities) 
using a given {\em scoring matrix}. 
Defining a hit through an additive contribution of several positions
is captured by a general formalism of {\em vector seeds} proposed in
\cite{BrejovaBrownVinarJCSS05}. On the other hand, it has been
understood
\cite{PatternHunter04,BuhlerRECOMB04,KucherovNoeRoytberg04,YangWangChenEtAlBIBE04,XuBrownLiMaCPM04}
that using simultaneously a {\em family} of seeds instead of a single
seed can further improve the sensitivity/selectivity ratio. Papers
\cite{BrownTCBB05,tPatternHunter05} both propose solutions using
a family of vector seeds to surpass the performance of {\sc Blast}. 

However, using the principle of cumulative score over several
adjacent positions has an algorithmic cost. 
Defining a hit through a pattern of exact letter matches
allows for a {\em direct hashing} scheme, where each key of the query
sequence is associated with a {\em unique} hash table entry pointing to
the positions of the subject sequence (database) where the key can hit. 
Usually these positions are stored in consecutive memory cells within
the hash table. 

On
the other hand, defining a hit through a cumulative contribution of
several positions leads to an additional pre-computed table that stores,
for each key, its {\em neighborhood} i.e., the list of subject keys
(or corresponding hash table entries) with which it can form a
hit. For example, in a standard {\sc Blastp} setting (Blosum62 scoring
matrix with threshold 11 for cumulative score of three 
positions), the expectation, computed according to the Bernoulli
sequence model, of the number of neighbors of a key is 
19.34, 
i.e. that many accesses to the hash table are required for each key. 
For four positions and threshold 18, as in the case of 
seeds from \cite{BrownTCBB05}, 
a key hits expectedly 15.99 keys
and this number grows up to 45.59 when the score threshold decreases to 16. 
This raises an obvious memory problem: for example, for key size 4 and
score threshold 18, the total size of neighborhoods is 7609575, and
for key size 5 the neighborhood table may simply not fit into the memory. 
Another related implementation problem is cache usage: different keys
of a neighborhood generally correspond to remote segments of the hash
table and their processing gives rise to cache misses that cause
additional latencies. 

Those implementation issues may become a bottleneck in large-scale
protein comparisons. Furthermore, solving these problems may be
very helpful in different specific experimental setups, such as in mapping
protein comparison algorithms to specialized computer architecture
(see e.g. \cite{PeterlongoEtAlPBC07,NguyenLavenierRIVF08}) where memory usage may be a
crucial issue. 

In \cite{KucherovNoeRoytberg06}, we proposed a new concept of 
{\em subset seeds} that can be viewed as an intermediate between
ordinary spaced seeds and vector seeds: subset seeds allow one to
distinguish between different types of mismatches (or matches) but
still treat seed positions independently rather than
cumulatively. Distinguishing different mismatches is not done by
scoring them, but by extending the seed alphabet such that each seed
letter specifies different sets of mismatches. For example, in the DNA
case it is beneficial to distinguish between transition mutations
({\tt A} $\leftrightarrow$ {\tt G}, {\tt C} $\leftrightarrow$ {\tt T})
and others (transversions)
\cite{NoeKucherovBMC04,ZhouStantonFloreaBMCBioinformatics08}. This
leads (at least in the case of {\em transitive} seed alphabets defined
in this paper) to the possibility of using the direct hashing. 

Since the protein alphabet is much larger than the one of DNA, subset
seeds provide a very attractive seeding option for protein
alignment. 
In this paper, we study the performance of subset seeds applied to
protein sequences and compare it to existing seeding techniques of
{\sc Blastp} and vector seeds. 

Note again that subset seeds are less expressive than {\sc Blast}
seeds or vector seeds in general, but in return, admit a more
efficient implementation. Besides treating positions independently,
subset seeds replace amino acid substitution scores by simply distinguishing different classes
of mismatches. Therefore, another way to state the motivation of this
work is to ask {whether scores are really necessary at the seeding
stage of protein alignment}. We will show that with a reasonable
level of precision the answer to this question is negative. 

In the paradigm of subset seeds, each seed letter specifies a set of
amino acid pairs matched by this letter. Therefore, a crucial question
is the design of an appropriate {\em seed alphabet}, which is one of
the main problems we study in this paper. {\em In fine}, the quality
of an alphabet is determined by the quality of the best seeds that can
be constructed over this alphabet. The latter is already a complex
optimization problem that is usually solved in practice by heuristic
methods. (For a formal analysis of seed design problem we refer to the
recent paper \cite{MaYaoAPBC08} and references therein.) The
problem of alphabet design studied in this paper presents an
additional complexity as it introduces an additional dimension
of the search space (set of possible alphabets), and additionally
requires a study of selectivity/sensitivity dependencies rather than simply
maximizing the sensitivity for a class of seeds with a given
selectivity. In this paper we propose several heuristic methods that
lead to the design of efficient seed alphabets and corresponding
seeds.

The paper is organized as follows. In
Section~\ref{section:preliminaries}, we introduce some probabilistic
notions we need to reason about seed
efficiency. Section~\ref{section:non-transitive-seed-alphabet}
introduces the first simple approach to design a seed alphabet, which,
however, does not lead to so-called {\em transitive} seeds, useful in
practice. Section~\ref{section:transitive-seed-alphabet} presents
three 
different approaches to designing transitive seed alphabets, based on
a pre-defined (Section~\ref{subsection:prebuilded-tree}) or newly
designed (Section~\ref{subsection:ab-initio-clustering}) hierarchical
clustering of amino acids, as well as on a non-hierarchical clustering
(Section~\ref{subsection:non-hierarchical}).
Section~\ref{section:experiments} describes comparative experiments
made with the designed seeds, obtained both on probabilistic models and
on different protein data banks.

\section{Preliminaries}
\label{section:preliminaries}

Throughout the paper, we denote $\Sigma=\{\mathtt{A, C, D, E, F, G, H, I,
  K, L, M, N, P, Q, R, S, T, V, W, Y}\}=\{a_i\}_{i=1..20}$ the alphabet of amino acids. 

In most general terms, a {\em (subset) seed letter} $\alpha$ is
defined as any symmetric and reflexive binary relation on 
$\Sigma$. Let $\B$ be a {\em seed alphabet}, i.e. a collection of 
subset seed letters. Then a {\em subset seed} $\pi=\alpha_1 \ldots
\alpha_k$ is a word over $\B$, where $k$ is called the {\em span} of $\pi$. $\pi$ defines a 
symmetric and reflexive binary relation on words of $\Sigma^k$ (called
{\em keys}): for 
$s_1,s_2\in \Sigma^k$, $s_1 \sim_\pi s_2$ iff $\forall i\in [1..k]$, we have
$\langle s_1[i],s_2[i]\rangle\in \alpha_i$. In this case, we say that
seed $\alpha$ {\em hits} the pair $s_1,s_2$. 

For practical reasons, we would like seed letters to define a 
{\em transitive} relation, in addition. This induces an equivalence
relation on keys, which is very convenient and allows for an efficient
indexing scheme (see Introduction). In this paper, we will be mainly
interested in transitive seed letters, but we will also study
the non-transitive case in order to see how restrictive the
transitivity condition is. 

The quality of a seed letter or of a seed is characterized by two main
parameters:  
{\em sensitivity} and {\em selectivity}. They are defined through
background and foreground probabilistic models of protein alignments.
Foreground probabilities are assumed to represent the distribution of
amino acids matches in proteins of interest, when two homologous
proteins are aligned together. Background probabilities, on the other
hand, represent the distribution of amino acid matches in 
{\em random alignments}, when two proteins are randomly aligned
together. 




In this paper, we restrict ourselves to Bernoulli models of proteins
and protein alignments, although some of
the results we will present can be extended to Markov models. 

Assume that we are given background probabilities
$\{b_1,\ldots,b_{20}\}$ of amino acids in protein sequences under
interest. The {\em background probability} of a seed letter
$\alpha$ is defined by $b(\alpha)=\sum_{(a_i, a_j)\in \alpha} b_i
b_j$. The {\em selectivity} of $\alpha$ is $1-b(\alpha)$ and the 
{\em weight} of $\alpha$ is defined by 
\begin{equation}
w(\alpha)=\frac{\log b(\alpha)}{\log b(\#)},
\end{equation}
where $\#=\{\langle a,a\rangle|a\in\Sigma\}$ is the
``identity'' seed letter.  
For a seed $\pi=\alpha_1 \ldots \alpha_k$, the background probability
of $\pi$ is $b(\pi)=\prod_{i=1}^k b(\alpha_i)$, the selectivity of
$\pi$ is $1-b(\pi)$ and the weight of $\pi$ is $w(\pi)=\log_{b(\#)}
  b(\pi)=\sum_{i=1}^k w(\alpha_i)$. 
%
Note that the weight here generalizes the weight of classical
spaced seeds \cite{KeichLiMaTrompDAM04} defined as the number of
``identity'' letters it contains. 

Let $f_{ij}$ be the probability to see the pair $\langle a_i,a_j\rangle$ 
aligned in a target alignment. The {\em foreground probability} of a
seed letter $\alpha$ is defined by $f(\alpha)=\sum_{(a_i, a_j)\in \alpha} f_{ij}$. 
The {\em sensitivity} of a seed $\pi$ is defined as
the probability to hit a
random target alignment\footnote{Note that our definitions of
  sensitivity and selectivity are not symmetric: sensitivity is
  defined with respect to the entire alignment and selectivity with respect to a
single alignment position. These definitions capture better the
intended parameters we want to measure. However, selectivity could
also be defined with respect to the entire alignment. We could suggest
the term {\em specificity} for this latter definition.}. 
Assume that target alignments are specified
by a length $N$. Then the sensitivity of a seed
$\pi=\alpha_1 \ldots\alpha_k$ is
the probability that a randomly drawn gapless alignment (i.e. string of
pairs $\langle a_i,a_j\rangle$) of length $N$ contains a fragment of length $k$
which is matched by $\pi$. In \cite{KucherovNoeRoytberg06} we proposed
a general algorithm to efficiently compute the seed sensitivity for a 
broad class of target alignment models. This algorithm will be used
in the experimental part of this work. 

The general problem of seed design is to obtain seeds with good
sensitivity/selectivity trade-off. Even within a fixed seed formalism, 
the quality of a seed is dependent on the chosen selectivity value. 
This is why we will always be interested in
computing efficient seeds 
for a large range of selectivity levels. 

\section{Dominating seed letters}
\label{section:non-transitive-seed-alphabet}

Our main question is how to choose seed letters that form good seeds?
Intuitively, ``good letters'' are those that best distinguish
foreground and background letter alignments. 

For each letter $\alpha$, consider its foreground and background
probabilities $f(\alpha)$ and $b(\alpha)$ respectively. 
Intuitively, we would like to have
letters $\alpha$ with large $f(\alpha)$ and small $b(\alpha)$. A
letter $\alpha$ is said to {\em dominate} a letter $\beta$ if $f(\alpha)\geq
f(\beta)$ and $b(\alpha)\leq b(\beta)$. Observe that in this case,
$\beta$ can be removed from consideration, as it can always be
advantageously replaced by $\alpha$. 


Consider all amino acid pairs $(a_i,a_j)$ ordered by descending 
{\em likelihood ratio} $f_{ij}/b_i b_j$. Consider the set of pairs
$R(t)=\{(a_i,a_j)\,|\,f_{ij}/b_i b_j>t\}$. Then 
the following statement holds\footnote{It is interesting to
point out the relationship to the Neyman-Pearson lemma which is a more
general formulation of this statement.}.
\begin{proposition}
\label{prop-dominating}
$R(t)$ cannot be dominated by any other letter.
\end{proposition}
\begin{proof}
Assume by contradiction that $R(t)$ is dominated by some letter
$\alpha$, i.e. $f(\alpha)\geq f(R(t))$ and $b(\alpha)\leq
b(R(t))$. Consider $\beta=R(t)\setminus \alpha$ and
$\gamma=\alpha\setminus R(t)$. Clearly, $f(\beta)\leq f(\gamma)$ and
$b(\beta)\geq b(\gamma)$. On the other hand, $\forall (a_i,a_j)\in
\beta$, $f_{ij}/b_ib_j>t$ and $\forall (a_i,a_j)\in
\gamma$, $f_{ij}/b_ib_j\leq t$. This implies that
$f(\beta)=\sum_{(a_i,a_j)\in\beta}f_{ij}>t\sum_{(a_i,a_j)\in\beta}b_{i}b_{j}=tb(\beta)$
and similarly $f(\gamma)\leq tb(\gamma)$. We then have
$f(\beta)>tb(\beta)\geq tb(\gamma)\geq f(\gamma)$
which contradicts $f(\beta)\leq f(\gamma)$. 
\end{proof}
Proposition~\ref{prop-dominating} suggests that letters $R(t)$ are
good candidates to be included to the seed alphabet.
%
\begin{figure}\center
\includegraphics[width=0.6\textwidth]{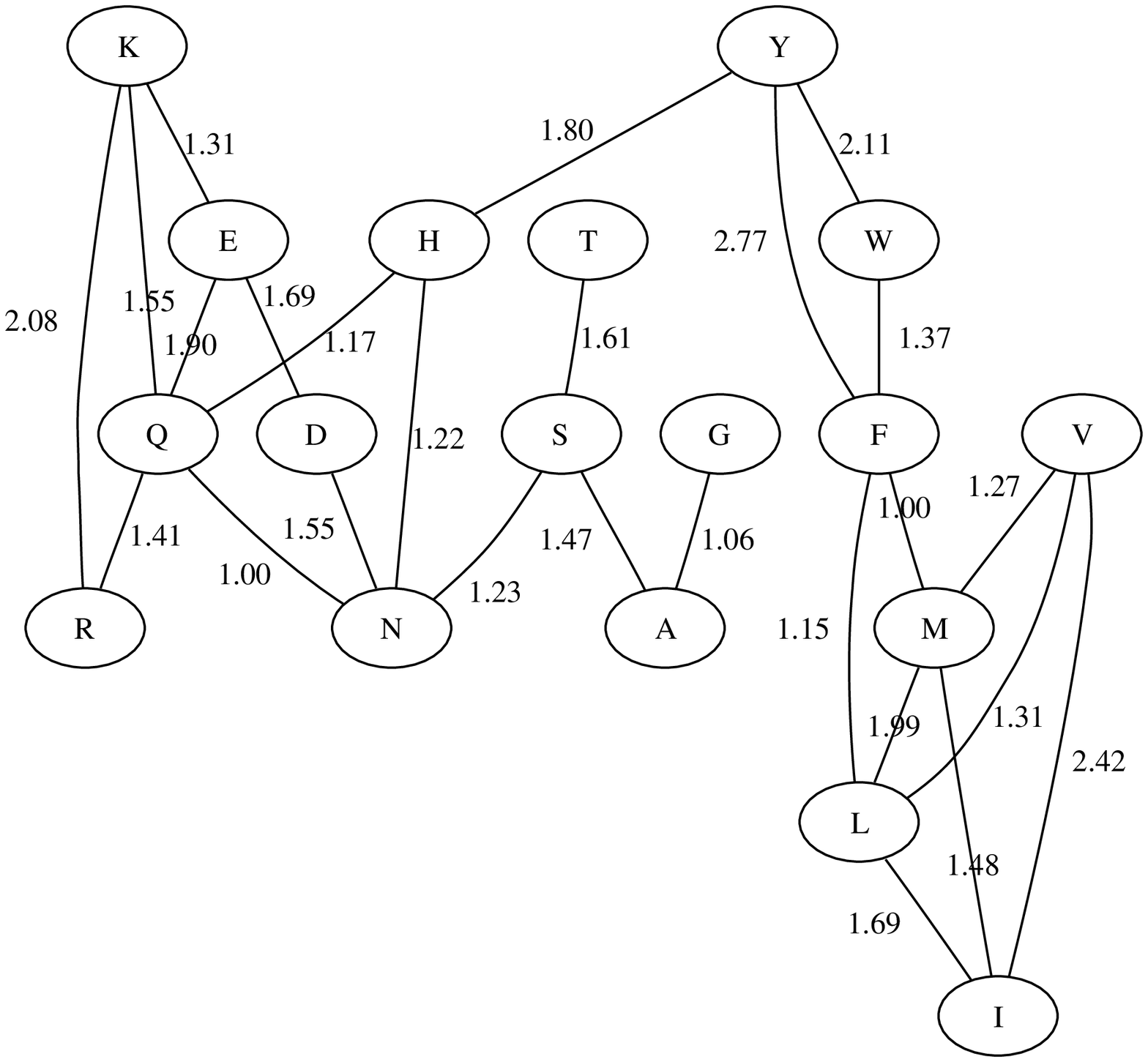}
\caption{Amino acid pairs forming letter $R(1)$
of alphabet {\nontrans}}
\label{figure:likelihood_threshold}
\end{figure}
%

\textit{Resulting alphabet.} 
We computed the likelihood ratio for all amino acid pairs, based on
practical values of 
background and foreground probabilities computed in
accordance with the BLOSUM62 matrix (see Section~\ref{subsection:methods}). 
Not surprisingly, amino acid identities (pairs $\langle a,a\rangle$)
have highest likelihood scores varying from $38.11$ for tryptophan (W)
down to $3.69$ for valine (V). 

Among non-identical pairs, only $25$ have a score greater than $1$
(Figure~\ref{figure:likelihood_threshold}). A quick
analysis shows that those do
not form a transitive relation, and therefore $R(1)$ does not verify the
transitivity requirement. This is also the case for other threshold
values. 

We analyzed a family of threshold letters
$R(t)$ for $t$ ranging from $0$ to $3$ with step $0.05$. At the
extremities of this interval, $R(0)$
is the ``joker'' letter admitting all amino acid pairs, and $R(3)$ is the
letter corresponding to the exact match relation. Among all those
letters,
there are only $34$ different ones. This alphabet of $34$ letters 
(data not shown),
denoted {\nontrans}, will be used in the experimental part of the paper 
(Section~\ref{section:experiments}) in order to study how restrictive
the requirement of 
transitive letters is, i.e. how much better are general seeds than those
obtained with the restriction of transitivity. 


\section{Transitive seed alphabets}
\label{section:transitive-seed-alphabet}
In the case of transitive seed alphabets, every letter
$\al\in\B$ is a partition of the amino acid alphabet $\Sigma$. 
In other words, the binary relation associated with each letter 
(cf Section~\ref{section:preliminaries}) is an equivalence relation. 
Transitive alphabets represent the practical case when each amino acid 
is uniquely mapped to its equivalence class. This, in turn, allows
for an efficient hashing scheme during the stage of seed search, when 
different entries of the hash table index non-intersecting subsets of keys. 


In
Sections~\ref{subsection:prebuilded-tree} and \ref{subsection:ab-initio-clustering} below,
we explore transitive seed alphabets satisfying an additional  
``hierarchy condition'': for any two seed letters $\al_1,\al_2\in\B$ corresponding to 
partitions $P_{\alpha_1},P_{\alpha_2}$ respectively, one of
$P_{\alpha_1},P_{\alpha_2}$ is a refinement of the other. Formally,
\begin{equation}
\mbox{for any }\al_1,\al_2\in\B, \mbox{ either } \al_1 \prec \al_2, \mbox{
  or } \al_2 \prec \al_1,
\label{equation:refinement}
\end{equation}
where $\al\prec\beta$ means that every set of $P_{\beta}$ is a subset
of some set of $P_{\alpha}$. 

The purpose of the above requirement is to define seed letters using a 
biologically significant hierarchical clustering of amino acids represented
by a tree. In Section~\ref{subsection:prebuilded-tree}, we will use a 
pre-defined hierarchical clustering to design efficient seed alphabets.
Then in Section~\ref{subsection:ab-initio-clustering}, we construct our
own clustering based on appropriate background and
  foreground models of amino acids distribution.
Finally, in Section~\ref{subsection:non-hierarchical} we lift condition
(\ref{equation:refinement}) and study ``non-hierarchical'' seed alphabets. 

\subsection{Transitive alphabets based on a pre-defined clustering}
\label{subsection:prebuilded-tree}
Assume we have a biologically significant hierarchical clustering tree which is a
rooted binary tree $T$ with $20$ leaves labeled by amino acids. Such
trees have been proposed in
\cite{LiFanWangWangJPE03,MurphyWallqvistLevyJPE00}, based on different
similarity relations. 
%
%
The hierarchical tree derived from \cite{LiFanWangWangJPE03} is shown in 
Figure~\ref{figure:tree1}. %
\begin{figure}
\includegraphics[height=.4\textheight,width=\textwidth]{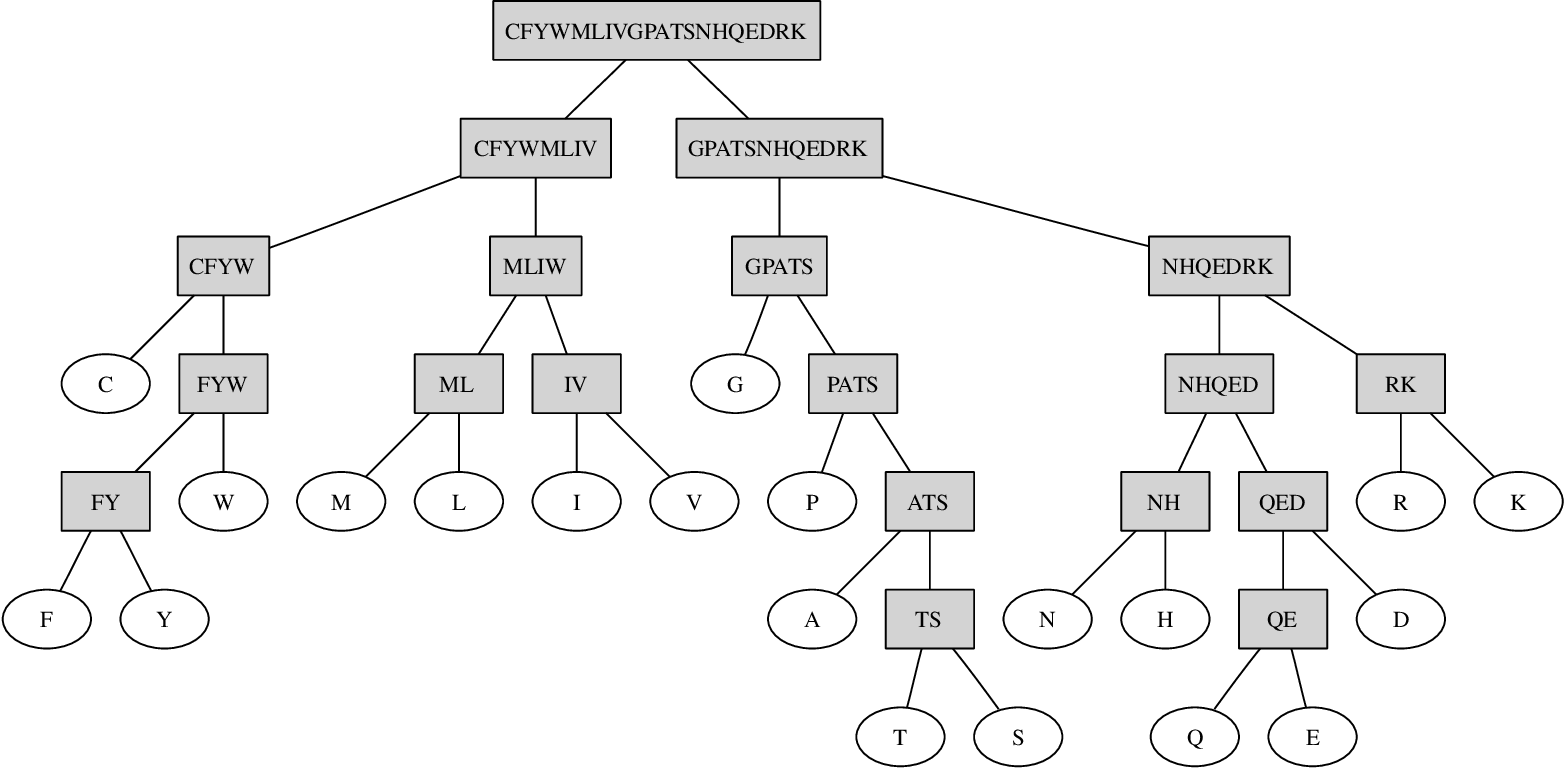}
\caption{Hierarchical tree derived from~\cite{LiFanWangWangJPE03}}
\label{figure:tree1}
\end{figure}
The tree, obtained with a purely bioinformatics analysis, groups
together amino acids with similar biochemical properties, such as
hydrophobic amino acids {\tt L,M,I,V}, hydrophobic aromatic amino
acids {\tt F,Y,W}, amino acids with an alcohol group {\tt S,T}, or charged/polar amino acids
{\tt E,D,N,Q}. A similar grouping has been obtained in
\cite{MurphyWallqvistLevyJPE00}. 

A {\em seed letter} is defined here as a subset $\alpha$ of nodes of $T$
such that
\begin{itemize}
\item[(i)] $\alpha$ contains all leaves,
\item[(ii)] for a node $v$, if $v\in\alpha$, then all descendants of
  $v$ belong to $\alpha$ too.
\end{itemize}
In other words, a seed letter can be thought of as a ``horizontal
cut'' of the tree. 
Clearly, each letter induces a partition on the set of leaves (amino
acids). For example, for the tree on Figure~\ref{figure:tree1}, a
letter defined by the cut through nodes {C}, {FYW}, {MLIV}, {G}, {P},
{ATS}, {NHQEDRK} corresponds to the partition {\tt
  \{\{C\},\{FYW\},\{MLIV\},\{G\},\{P\},\{ATS\},\{NHQEDRK\}\}}. 

Seed letters are naturally ordered by inclusion.
The smallest one is the ``identity'' seed letter $\#$,
containing only the leaves.
The largest one is the ``joker'' seed letter $\_\,$,
containing all the nodes of $T$.
One particular seed letter is obtained by removing from
$\_$ the root node. We denote it by $@$.

Observe that each seed letter $\alpha$ represents naturally
an equivalence relation on $\Sigma$:
$a_i$ and $a_j$ are related iff their common ancestor belongs to
$\alpha$. It is identity relation in case of $\#$ and
full relation in case of $\_\,$.



Following condition (\ref{equation:refinement}), a \emph{hierarchical seed alphabet} is a family $\G$ of seed letters
such that 
\begin{equation}
\mbox{for every }\al_1, \al_2 \in \G,
\mbox{ either }\al_1 \subseteq \al_2\mbox{ or }\al_2 \subseteq \al_1. 
\label{equation:consistent}
\end{equation}
Hence, in mathematical terms, a seed alphabet is a chain in the inclusion
ordering of seed letters. Each hierarchical alphabet can be obtained
by a series of refinements (set splittings) of its least refined
letter. 

Let us analyze what are the maximal seed alphabets
wrt.\ inclusion. Clearly each maximal seed alphabet $\G$ always
contains the smallest and the largest seed letters $\#$ 
and $\_$\ .
Interestingly, each maximal alphabet $\G$ contains also $@$,
as $@$ is comparable (by inclusion) to any other seed letter. 

It can be shown that under the above definitions, any maximal seed alphabet contains exactly $20$
letters that can be obtained by a stepwise merging of two subtrees
rooted at immediate descendants of some node $v$ into the subtree
rooted at $v$. Therefore, since a binary tree with $n$ leaves contains
$n-1$ internal nodes, a maximal seed alphabet contains precisely $20$
letters and can be specified by a permutation of internal nodes in
tree $T$. 

\textit{Seed alphabets and constraint independence systems.}
It is interesting to observe that the set of seed alphabets forms a
{\em constrained independence system} \cite{triangul}. An independence
system is a collection of subsets $ \II \subseteq 2^E$ over a ground
set $E$, called {\em independent sets}, 
such that {\em (i)} $\emptyset \in \II$, and {\em (ii)} if $X \in \II$ and
$Y \subseteq X$, then $Y \in \II$. A maximal (w.r.t. inclusion) independent 
set is called a {\em base}. 

Let $E$ be the set of all possible seed letters as
defined earlier. Then alphabets 
verifying (\ref{equation:consistent}) form an independence system, where bases
correspond to maximal seed alphabets. Moreover, seed alphabets verify
two additional conditions of {\em constrained independent system} \cite{triangul}:
{\em (iii)} if $X, Y \in \II$ with $|Y| < |X|$, then there is an element
$e \in E \setminus Y$ such that $Y \cup \{e\}\in\II$, and {\em (iv)} the
cardinality of every minimal (w.r.t inclusion) set of $2^E\setminus\II$ is two. 

The interest of this observation follows from results of
\cite{triangul} showing that some optimization problems on constrained
independence systems can be solved efficiently by greedy algorithms. 
Assume we have a score function $s : E \to R$
that we extend additively to independent sets by
$s(X)=\sum_{e\in X}s(e)$. For an independence system $\II$, we want to 
find a base $X\in\II$ with optimal (maximal or minimal) $s(X)$. 
For constrained independence systems,
it was proved~\cite{triangul}  that the greedy algorithm
yields a base which is {\em locally  optimal}, i.e. better than any
neighbor base 
$Y=(X\setminus\{\alpha_1\})\cup\{\alpha_2\}$ for some $\alpha_1\in X$,
$\alpha_2\in E\setminus X$. Here, the greedy algorithm starts with the
empty set and iteratively adds most optimal elements of $E$ as long as
the current set remains independent. The absolute optimum is hard to
compute in general, and the greedy solution is an
approximation of it. 

\textit{Assigning letter score.}
The above setting requires that each letter $\alpha$ is assigned a
score that, intuitively, should measure the ``usefulness'' of
$\alpha$ in a potential alphabet. Defining such a measure is a
difficult question as there are too many potential alphabets and
we can not check them all exhaustively. Therefore, we chose to
consider only small alphabets $\G_{\al}$, containing $\al$ together
with a few other letters that are always present in a good seed
alphabet. Those letters are $\{\_, @, \# \}$. 
The experiments reported in
Section~\ref{section:experiments} use the alphabet 
$\G_{\al}=\{\_,\al\}$. 

\begin{figure}[b!]
{\tiny
\begin{eqnarray*}
& \{CFYWMLIVGPATSNHQEDRK\} \\
& \{CFYWMLIV\}\;        \{GPATSNHQEDRK\}\;\\
& \{CFYWMLIV\}\;        \{GPATS\}\;   \{NHQEDRK\}\;\\
& \{CFYW\}\;    \{MLIV\}\;    \{GPATS\}\;   \{NHQEDRK\}\;\\
& \{CFYW\}\;    \{MLIV\}\;    \{G\}\;       \{PATS\}\;    \{NHQEDRK\}\; \\
& \{C\}\;       \{FYW\}\;     \{MLIV\}\;    \{G\}\;       \{PATS\}\;    \{NHQEDRK\}\; \\
& \{C\}\;       \{FYW\}\;     \{MLIV\}\;    \{G\}\;       \{P\}\;       \{ATS\}\;     \{NHQEDRK\}\; \\
& \{C\}\;       \{FY\}\;      \{W\}\;       \{MLIV\}\;    \{G\}\;       \{P\}\;       \{ATS\}\;     \{NHQEDRK\}\;   \\     
& \{C\}\;       \{F\}\;       \{Y\}\;       \{W\}\;       \{MLIV\}\;    \{G\}\;       \{P\}\;       \{ATS\}\;     \{NHQEDRK\}\;     \\          
& \{C\}\;       \{F\}\;       \{Y\}\;       \{W\}\;       \{MLIV\}\;    \{G\}\;       \{P\}\;       \{A\}\;        \{TS\}\;      \{NHQEDRK\}\; \\
& \{C\}\;       \{F\}\;       \{Y\}\;       \{W\}\;       \{MLIV\}\;    \{G\}\;       \{P\}\;       \{A\}\;       \{T\}\;       \{S\}\;       \{NHQEDRK\}\;\\
& \{C\}\;       \{F\}\;       \{Y\}\;       \{W\}\;       \{MLIV\}\;    \{G\}\;       \{P\}\;       \{A\}\;       \{T\}\;       \{S\}\;       \{NHQED\}\;   \{RK\}\;  \\    
& \{C\}\;       \{F\}\;       \{Y\}\;       \{W\}\;       \{MLIV\}\;    \{G\}\;       \{P\}\;       \{A\}\;       \{T\}\;       \{S\}\;       \{NHQED\}\;   \{R\}\;       \{K\}\;    \\   
& \{C\}\;       \{F\}\;       \{Y\}\;       \{W\}\;       \{MLIV\}\;    \{G\}\;       \{P\}\;       \{A\}\;       \{T\}\;       \{S\}\;       \{NH\}\;      \{QED\}\;     \{R\}\;       \{K\}\;       \\
& \{C\}\;       \{F\}\;       \{Y\}\;       \{W\}\;       \{MLIV\}\;    \{G\}\;       \{P\}\;       \{A\}\;       \{T\}\;       \{S\}\;       \{N\}\;       \{H\}\;       \{QED\}\;     \{R\}\;       \{K\}\;    \\
& \{C\}\;       \{F\}\;       \{Y\}\;       \{W\}\;       \{MLIV\}\;    \{G\}\;       \{P\}\;       \{A\}\;       \{T\}\;       \{S\}\;       \{N\}\;       \{H\}\;       \{QE\}\;      \{D\}\;       \{R\}\;       \{K\}\;       \\
& \{C\}\;       \{F\}\;       \{Y\}\;       \{W\}\;       \{MLIV\}\;    \{G\}\;       \{P\}\;       \{A\}\;       \{T\}\;       \{S\}\;       \{N\}\;       \{H\}\;       \{Q\}\;       \{E\}\;       \{D\}\;       \{R\}\;       \{K\}\;       \\
& \{C\}\;       \{F\}\;       \{Y\}\;       \{W\}\;       \{ML\}\;      \{IV\}\;      \{G\}\;       \{P\}\;       \{A\}\;       \{T\}\;       \{S\}\;       \{N\}\;       \{H\}\;       \{Q\}\;       \{E\}\;       \{D\}\;       \{R\}\;       \{K\}\;       \\
& \{C\}\;       \{F\}\;       \{Y\}\;       \{W\}\;       \{M\}\;       \{L\}\;       \{IV\}\;      \{G\}\;       \{P\}\;       \{A\}\;       \{T\}\;       \{S\}\;       \{N\}\;       \{H\}\;       \{Q\}\;       \{E\}\;       \{D\}\;       \{R\}\;       \{K\}\;       \\
& \{C\}\;       \{F\}\;       \{Y\}\;       \{W\}\;       \{M\}\;       \{L\}\;       \{I\}\;       \{V\}\;       \{G\}\;       \{P\}\;       \{A\}\;       \{T\}\;       \{S\}\;       \{N\}\;       \{H\}\;       \{Q\}\;       \{E\}\;       \{D\}\;       \{R\}\;       \{K\}\;     
\end{eqnarray*}
}
\caption[Alphabet based on
  Figure~\ref{figure:tree1}]{Alphabet {\polish} designed using the tree of
  Figure~\ref{figure:tree1}. Each line corresponds to a seed letter (amino acid partition)}
\label{figure:pre-defined-alphabet}
\end{figure}
Given $\G_{\al}$, we define the score of $\al$ as follows. We
enumerate all seeds of a given span (typically, $5$ or $6$) over
$\G_{\al}$, and compute the sensitivity and
selectivity of each seed according to the protocol described in
Section~\ref{seed-design}. 
Each seed is then associated with a
point on a unit square with coordinates corresponding to sensitivity
and selectivity (see plots in Figure~\ref{ROC-curves} below). 
The distance of this point to point $(1, 1)$, denoted $\rho(\al)$, measures
how good the sensitivity and selectivity jointly are. Besides, the
number of occurrences of $\al$ in the seed should be taken into
account. Overall, we chose to compute the score of a letter by the
following formula: 
$$
w(\alpha)=\sum_\pi occ_{\pi}(\al)\cdot (\sqrt{2} - \rho(\al)),
$$
where the sum is taken over all seeds $\pi$ of a given span, and
$occ_{\pi}(\al)$ is the number of occurrences of $\al$ in $\pi$. 

\textit{Greedy algorithm.} 
Once every seed letter has been assigned a score, we compute the greedy
solution as follows. We compute the maximal
subset $L$ of {\em locally good letters}, i.e. letters
$\alpha$ that score better than any letter $\alpha'$ such that
$\{\alpha,\alpha'\}\not\in\II$. It can be shown that this subset is
independent and is included in the solution
computed by the greedy algorithm. Then we redefine $E$ and $\II$ by
$E'= E\setminus L$ and $\II' = \{Z\subseteq E'\,|\,Z\cup L\in\II\}$ and
apply the algorithm recursively to the independence system
$(E',\II')$. The union of all sets $L$ of locally good letters
computed along this procedure forms the solution of the greedy
algorithm. 

\textit{Resulting alphabet.}
Figure~\ref{figure:pre-defined-alphabet} shows alphabet 
{\polish} designed 
through the approach of this Section. 
The alphabet has been designed from the tree of
Figure~\ref{figure:tree1} and using the alphabet $\G_{\al}=\{\_,\al\}$ for
assigning the score of a letter $\alpha$. 
Each line in Figure~\ref{figure:tree1} corresponds to a
letter (amino acid partition). Among alphabets obtained by varying
different parameters in scoring individual letters 
(such as the alphabet and seed spans used in the scoring procedure), alphabet {\polish} produced best seeds and will be used in the
experimental part of this work (Section~\ref{section:experiments}). 

\subsection{Transitive alphabets using an {\em ab initio} clustering method}
\label{subsection:ab-initio-clustering}

\textit{Hierarchical clustering of amino acids.}
\label{subsubsection:ab-initio-clustering-algorithm}
A prerequisite to the approach of
Section~\ref{subsection:prebuilded-tree} is a given tree
describing a hierarchical clustering of amino acid based on some
similarity measure. In this section, we describe an 
approach that constructs {\em ab initio} a hierarchical clustering of
amino acids, using a likelihood measure. 
The approach can be seen as constructing a hierarchy of connected
components of a graph based on the likelihood relation considered in
Section~\ref{section:non-transitive-seed-alphabet} (see
Figure~\ref{figure:likelihood_threshold}) trying to build components
with high likelihood values.

As in Section~\ref{subsection:prebuilded-tree}, our goal here is
to construct a family of seed letters verifying the hierarchy condition
(\ref{equation:refinement}). 
This family will be obtained with a
simple greedy neighbor-joining clustering algorithm.
%
%
We start with the partition of amino acids into 20 singletons. This
partition corresponds to the $\#$ letter. For a current partition 
$P = \{C_1,\ldots,C_n\}$, iteratively apply the following procedure.
\begin{enumerate}
\item[1] For each pair of sets $C_k$, $C_\ell$,
  \begin{enumerate}
  \item[1.1] consider the set 
    $Bridge(C_k, C_\ell) = \{(a_i,a_j) | a_i \in C_k,\ a_j \in C_\ell\}$.
  \item[1.2] compute $ForeBridgeProb(k,\ell)=\sum\{f_{ij}|a_i\in C_k,\ a_j
    \in C_\ell\}$ and \linebreak[4]$BackBridgeProb(k,\ell)=\sum\{b_{i}b_j|a_i\in C_k,\ a_j
    \in C_\ell\}$,
  \item[1.3] compute $L(k,\ell) =  ForeBridgeProb(k,\ell) / BackBridgeProb(k,\ell)$ 
  \end{enumerate}
\item[2] Find the pair of sets $(C_k, C_\ell)$ yielding the maximal
  $L(k,\ell)$,
\item[3] Merge $C_k$ and $C_\ell$ into a new set, obtaining a new
  partition. 
\end{enumerate}

The rationale behind this simple procedure is that those two sets of
amino acids are merged together which produce the maximal increment in
the likelihood. 
An alternative method, when the likelihood of the whole
resulting set is maximized, yields biased results, as sets with a high
likelihood tend to ``absorb'' other sets. 

\begin{figure}[t!]\centering
{\tiny
\begin{eqnarray*}
& \{CFYWHMLIVPGQERKNDATS\}\; \\
& \{CFYWHMLIV\}\; \{PGQERKNDATS\}\; \\
& \{C\}\; \{FYWHMLIV\}\; \{PGQERKNDATS\}\; \\
& \{C\}\; \{FYWHMLIV\}\; \{P\}\; \{GQERKNDATS\}\; \\
& \{C\}\; \{FYWH\}\; \{MLIV\}\; \{P\}\; \{GQERKNDATS\}\; \\
& \{C\}\; \{FYWH\}\; \{MLIV\}\; \{P\}\; \{GATS\}\; \{QERKND\}\; \\
& \{C\}\; \{FYWH\}\; \{MLIV\}\; \{P\}\; \{G\}\; \{ATS\}\; \{QERKND\}\; \\
& \{C\}\; \{FYWH\}\; \{MLIV\}\; \{P\}\; \{G\}\; \{ATS\}\; \{QERK\}\; \{ND\}\; \\
& \{C\}\; \{FYW\}\; \{H\}\; \{MLIV\}\; \{P\}\; \{G\}\; \{ATS\}\; \{QERK\}\; \{ND\}\; \\
& \{C\}\; \{FYW\}\; \{H\}\; \{MLIV\}\; \{P\}\; \{G\}\; \{A\}\; \{TS\}\; \{QERK\}\; \{ND\}\; \\
& \{C\}\; \{FYW\}\; \{H\}\; \{MLIV\}\; \{P\}\; \{G\}\; \{A\}\; \{TS\}\; \{QE\}\; \{RK\}\; \{ND\}\; \\
& \{C\}\; \{FYW\}\; \{H\}\; \{ML\}\; \{IV\}\; \{P\}\; \{G\}\; \{A\}\; \{TS\}\; \{QE\}\; \{RK\}\; \{ND\}\; \\
& \{C\}\; \{FYW\}\; \{H\}\; \{ML\}\; \{IV\}\; \{P\}\; \{G\}\; \{A\}\; \{TS\}\; \{QE\}\; \{RK\}\; \{N\}\; \{D\}\; \\
& \{C\}\; \{FYW\}\; \{H\}\; \{ML\}\; \{IV\}\; \{P\}\; \{G\}\; \{A\}\; \{T\}\; \{S\}\; \{QE\}\; \{RK\}\; \{N\}\; \{D\}\; \\
& \{C\}\; \{FY\}\; \{W\}\; \{H\}\; \{ML\}\; \{IV\}\; \{P\}\; \{G\}\; \{A\}\; \{T\}\; \{S\}\; \{QE\}\; \{RK\}\; \{N\}\; \{D\}\; \\
& \{C\}\; \{FY\}\; \{W\}\; \{H\}\; \{ML\}\; \{IV\}\; \{P\}\; \{G\}\; \{A\}\; \{T\}\; \{S\}\; \{Q\}\; \{E\}\; \{RK\}\; \{N\}\; \{D\}\; \\
& \{C\}\; \{FY\}\; \{W\}\; \{H\}\; \{M\}\; \{L\}\; \{IV\}\; \{P\}\; \{G\}\; \{A\}\; \{T\}\; \{S\}\; \{Q\}\; \{E\}\; \{RK\}\; \{N\}\; \{D\}\; \\
& \{C\}\; \{FY\}\; \{W\}\; \{H\}\; \{M\}\; \{L\}\; \{I\}\; \{V\}\; \{P\}\; \{G\}\; \{A\}\; \{T\}\; \{S\}\; \{Q\}\; \{E\}\; \{RK\}\; \{N\}\; \{D\}\; \\
& \{C\}\; \{F\}\; \{Y\}\; \{W\}\; \{H\}\; \{M\}\; \{L\}\; \{I\}\; \{V\}\; \{P\}\; \{G\}\; \{A\}\; \{T\}\; \{S\}\; \{Q\}\; \{E\}\; \{RK\}\; \{N\}\; \{D\}\; \\
& \{C\}\; \{F\}\; \{Y\}\; \{W\}\; \{H\}\; \{M\}\; \{L\}\; \{I\}\; \{V\}\; \{P\}\; \{G\}\; \{A\}\; \{T\}\; \{S\}\; \{Q\}\; \{E\}\; \{R\}\; \{K\}\; \{N\}\; \{D\}\;
\end{eqnarray*}
}
 \caption{Alphabet {\russianT} obtained with the method of Section
 \ref{subsubsection:ab-initio-clustering-algorithm}}  
 \label{figure:amino_acids_neigboor_joining}
\end{figure}
\textit{Resulting alphabet.}
\label{subsubsection:ab-initio-clustering-example}
An alphabet, called {\russianT}, obtained with this greedy neighbor-joining approach
is given in
Figure~\ref{figure:amino_acids_neigboor_joining}. It will be used in
experiments presented later in Section~\ref{section:experiments}. 

\subsection{Non-hierarchical alphabets}
\label{subsection:non-hierarchical}

Previous approaches (Sections~4.1 and 4.2) were based on 
requirement (\ref{equation:refinement}) specifying that letters of the seed
alphabet should be embedded one into another to form a ``nested''
hierarchy. This requirement is biologically motivated and, on the
other hand, computationally useful as it reduces considerably the
space of possible letters. However, this requirement is not necessary
to implement the direct indexing (see Introduction). Therefore, we
also designed non-hierarchical alphabets in order to compare them to
hierarchical ones. 

To design non-hierarchical alphabets, we used a heuristic
generalizing the one of Section~\ref{subsection:ab-initio-clustering}. 
The heuristic consists of two stages: first, generate a big number
(several thousands)
of ``reasonable''  candidate letters, and then select from them
an alphabet containing ${\sim}20$ transitive letters (not necessarily
forming a hierarchy). 

The algorithm of the first stage exploits the standard paradigm of
genetic algorithms: it consequently creates ``generations'' of
transitive letters. The initial population  consists of a single
``identity'' seed letter. At the $k$-th iteration ($k = 1, \ldots,
19$), each letter generates $p$ descendants, each having $(20-k)$
sets. 

To generate descendants of a letter from the $k$-th generation,
we use the algorithm given in
Section~\ref{subsubsection:ab-initio-clustering-algorithm} but
maintain $p$ (instead of just 
one) best partitions according to the likelihood of the ``bridge''. The $(k+1)$-th
generation is selected among all descendants of the $k$-th generation
by selecting those $q$ letters $\alpha$ which have the highest likelihood ratio
$f(\alpha)/b(\alpha)$. With $p=100$ and $q = 500$ the
procedure gives about $8000$ candidate letters. 

To select a small number of those letters to form an alphabet, we tried
different heuristics based on the following two ideas: (1) letters with high
likelihood ratio are preferred (2) alphabet letters should have a
range of different weights. The second option produced a better
alphabet.

\textit{Resulting alphabet.}
We selected twenty letters out of about $8000$ candidates by
partitioning the candidates into twenty groups according to their
weight ranging from $0$ to $1$ with increment $0.05$, and by picking
in each group the letter with maximal likelihood. 
An alphabet obtained with the above heuristic, called {\russianNT}, 
is shown in
Figure~\ref{alph-non-hierar}. 
This alphabet will be used in the experiments reported in
Section~\ref{section:experiments}. 

\begin{figure}[h!]
{\tiny
\begin{eqnarray*}
& \{ARNDCQEGHILMKFPSTWYV\}\;\\
& \{ARNDQEGHILMKFPSTWYV\}\; \{C\}\;\\ 
& \{ARNDCQEHILMKFPSTWYV\}\; \{G\}\;\\ 
& \{ARNDQEHILMKFSTYV\}\; \{CGPW\}\;\\
& \{ARCQEHILMKFSTYV\}\; \{NDGPW\}\;\\
& \{ARNDCQEGHKPST\}\; \{ILMFWYV\}\;\\
& \{ARNDQEGHKST\}\; \{CILMFWYV\}\; \{P\}\;\\
& \{ARNDQEHKPST\}\; \{CW\}\; \{G\}\; \{ILMFYV\}\;\\
& \{ARNDQEKST\}\; \{CP\}\; \{GHW\}\; \{ILMFYV\}\;\\
& \{AGPST\}\; \{RNDQEHK\}\; \{C\}\; \{ILMFWYV\}\;\\
& \{APST\}\; \{RNDQEHK\}\; \{CW\}\; \{G\}\; \{ILMFYV\}\;\\
& \{AGST\}\; \{RNDQEK\}\; \{C\}\; \{HFWY\}\; \{ILMV\}\; \{P\}\;\\
& \{AST\}\; \{RNDQEK\}\; \{CH\}\; \{G\}\; \{ILMV\}\; \{FWY\}\; \{P\}\;\\
& \{AST\}\; \{RQEHK\}\; \{ND\}\; \{CP\}\; \{G\}\; \{ILMV\}\; \{FWY\}\;\\
& \{AST\}\; \{RQK\}\; \{NH\}\; \{DE\}\; \{C\}\; \{G\}\; \{ILMV\}\; \{FWY\}\; \{P\}\;\\
& \{A\}\; \{RQK\}\; \{N\}\; \{DE\}\; \{C\}\; \{G\}\; \{H\}\; \{ILMV\}\; \{FY\}\; \{P\}\; \{ST\}\; \{W\}\;\\
& \{A\}\; \{RK\}\; \{N\}\; \{DE\}\; \{C\}\; \{QH\}\; \{G\}\; \{ILV\}\; \{M\}\; \{FY\}\; \{P\}\; \{ST\}\; \{W\}\;\\
& \{A\}\; \{RQK\}\; \{ND\}\; \{C\}\; \{E\}\; \{G\}\; \{H\}\; \{IV\}\; \{LM\}\; \{FWY\}\; \{P\}\; \{ST\}\;\\
& \{A\}\; \{RK\}\; \{ND\}\; \{C\}\; \{Q\}\; \{E\}\; \{G\}\; \{H\}\; \{IV\}\; \{LM\}\; \{FWY\}\; \{P\}\; \{S\}\; \{T\}\;\\
& \{A\}\; \{RK\}\; \{N\}\; \{D\}\; \{C\}\; \{Q\}\; \{E\}\; \{G\}\; \{H\}\; \{IV\}\; \{L\}\; \{M\}\; \{FY\}\; \{P\}\; \{S\}\; \{T\}\; \{W\}\;\\
& \{A\}\; \{R\}\; \{N\}\; \{D\}\; \{C\}\; \{QE\}\; \{G\}\; \{H\}\; \{I\}\; \{L\}\; \{K\}\; \{M\}\; \{FWY\}\; \{P\}\; \{S\}\; \{T\}\; \{V\}\;\\
& \{A\}\; \{R\}\; \{N\}\; \{D\}\; \{C\}\; \{Q\}\; \{E\}\; \{G\}\; \{H\}\; \{I\}\; \{L\}\; \{K\}\; \{M\}\; \{F\}\; \{P\}\; \{S\}\; \{T\}\; \{W\}\; \{V\}\;
\end{eqnarray*}
}
\caption[Non-hierarchical alphabet]{Non-hierarchical alphabet {\russianNT} designed with the 
algorithm of Section~\ref{subsection:non-hierarchical}. 
}
\label{alph-non-hierar}
\end{figure}

\section{Experiments}
\label{section:experiments}
This section describes the experiments we made to
test the efficiency of seeds we designed with different
methods of previous sections. 
Sections~\ref{subsection:methods} - \ref{blast-evaluation} describe
the experimental protocol, from the assignment of background
and foreground probabilities, to the seed design. In
Section~\ref{subsection:results_theoretical}, we analyze the power of
different seed models proposed in
Sections~\ref{section:non-transitive-seed-alphabet}-\ref{section:transitive-seed-alphabet}
with respect to probabilistic models. 
Then in Section~\ref{subsection:results_real_data}, we benchmark the
performance of seeds built over different alphabets from
Section~\ref{section:transitive-seed-alphabet} against {\sc Blastp}, on
several reference protein databases. 
For
Sections~\ref{subsection:results_theoretical} and \ref{subsection:results_real_data},
all relative experimental data including scripts, designed alphabets, seeds and seed
families, and resulting sensitivity and selectivity measures, have
been collected in a supplementary Web page available at
{\tt\small\url{http://bioinfo.lifl.fr/yass/iedera_proteins/}}.

\subsection{Probability assignment and alphabet generation}
\label{subsection:methods}
First of all, we derived probabilistic models in accordance with the
BLOSUM62 data from the original paper \cite{BLOSSUM92}. We obtained the
BLOCKS database (version 5) \cite{BLOCKS91} and the software of
\cite{BLOSSUM92} to infer Bernoulli probabilities for the background
and foreground alignment models. These probabilities have been used
throughout the whole pipeline of experiments. 

Different seed alphabets have then been generated by the methods 
presented in Section~\ref{section:non-transitive-seed-alphabet}
(alphabet \nontrans),
Section~\ref{subsection:prebuilded-tree} (alphabet \polish),
Section~\ref{subsection:ab-initio-clustering} (alphabet \russianT) and
Section~\ref{subsection:non-hierarchical} (alphabet \linebreak[4]\russianNT). 

\subsection{Seed design}
\label{seed-design}

To each alphabet, we applied a seed design procedure that we briefly
describe now. 
Since each seed (or seed family) is characterized by two parameters 
-- sensitivity and selectivity -- 
it can be associated with a point on a 2-dimensional 
plot. Best seeds are then defined to be those which belong to
the {\em Pareto} set among all seeds, i.e. those than cannot be
strictly improved by increasing sensitivity, selectivity, or both. 

For different selectivity levels, we designed good seed families
containing one to six individual seeds, among which the best family was
selected. In each seed family, individual seeds have been
chosen to have approximately the same weight, within 5\%
tolerance. This requirement is natural as in the case of
divergent weights, seeds with lower weight would dominantly affect the
performance. In practice, having individual seeds of similar
weight allows an efficient parallel implementation
(see e.g. \cite{PeterlongoEtAlPBC07}). 

Estimation of sensitivity of individual seeds or seed families has
been done with the algorithm described in
\cite{KucherovNoeRoytberg06} and implemented in the {\sc Iedera}
software, available at {\tt {\small
    \url{http://bioinfo.lifl.fr/yass/iedera.php}}}. 
The selectivity of an individual seed has been computed according to
the definition (Section~\ref{section:preliminaries}). For a seed
family, its selectivity has been lower-estimated by summing the
background probabilities of individual seeds. 

Seed family design has been done using a hill climbing heuristic (see
\cite{BuhlerKeichSunRECOMB03,IlieIlieBIOCOMP07}) alternating seed
generation and seed estimation steps. 
All experiments were conducted for alignment lengths 16 and 32. 

\subsection{{\sc Blastp} and the vector seed family from \cite{BrownTCBB05}}
\label{blast-evaluation}

Our goal is to compare between different seed design approaches
proposed in this paper, but also to benchmark them against other
reference seeding methods. We used two references: the {\sc Blastp} 
seeding method and the family of vector seeds proposed in
\cite{BrownTCBB05}. Both of them use a score (or weight) resulting
from the cumulative contribution of several neighboring positions to
define a hit (see Introduction). Therefore, they use a more powerful
(and also more costly to implement) formalism of seeding. 

To estimate the sensitivity and selectivity of those seeds, we
modified our methods described in the previous section by representing
an alignment by a sequence of possible individual scores. Foreground
and background probability of each score is easily computed from those
for amino acid pairs. After that,
sensitivity and selectivity is computed similarly to the previous
case. 

\subsection{Results on theoretical models}
\label{subsection:results_theoretical}

We compare the performance of the different approaches by
plotting ROC curves of Pareto-optimal sets of seeds on the
selectivity/sensitivity graph. The two plots in Figure
\ref{figure:ROC-theoretical-zoom} show the results for alignment
length 16 and 32 respectively. Red and green polylines show the 
performance of {\sc Blastp} with word size 3 and the vector seed
family from \cite{BrownTCBB05}, for different score thresholds. The
other curves show the performances of different seed alphabets
from
Sections~\ref{section:non-transitive-seed-alphabet}-\ref{section:transitive-seed-alphabet}
represented by the Pareto-optimal seeds (seed families) that we were able to
construct over those alphabets. As mentioned earlier in
Section~\ref{seed-design}, each time we selected the best seed family
among those with different number of individual seeds. Typically (but
not exclusively), points on the plots correspond to seed families with
4 to 6 seeds.
Typically, the seed span ranges between 3 and 5 (respectively, 3 and 6) 
for alignment length 16 (respectively, 32). 
Seeds with larger span ($>4$) tend to occur in seed families with larger
number of seeds ($>3$). 
\begin{figure}[h!]\center
\includegraphics[width=12cm]{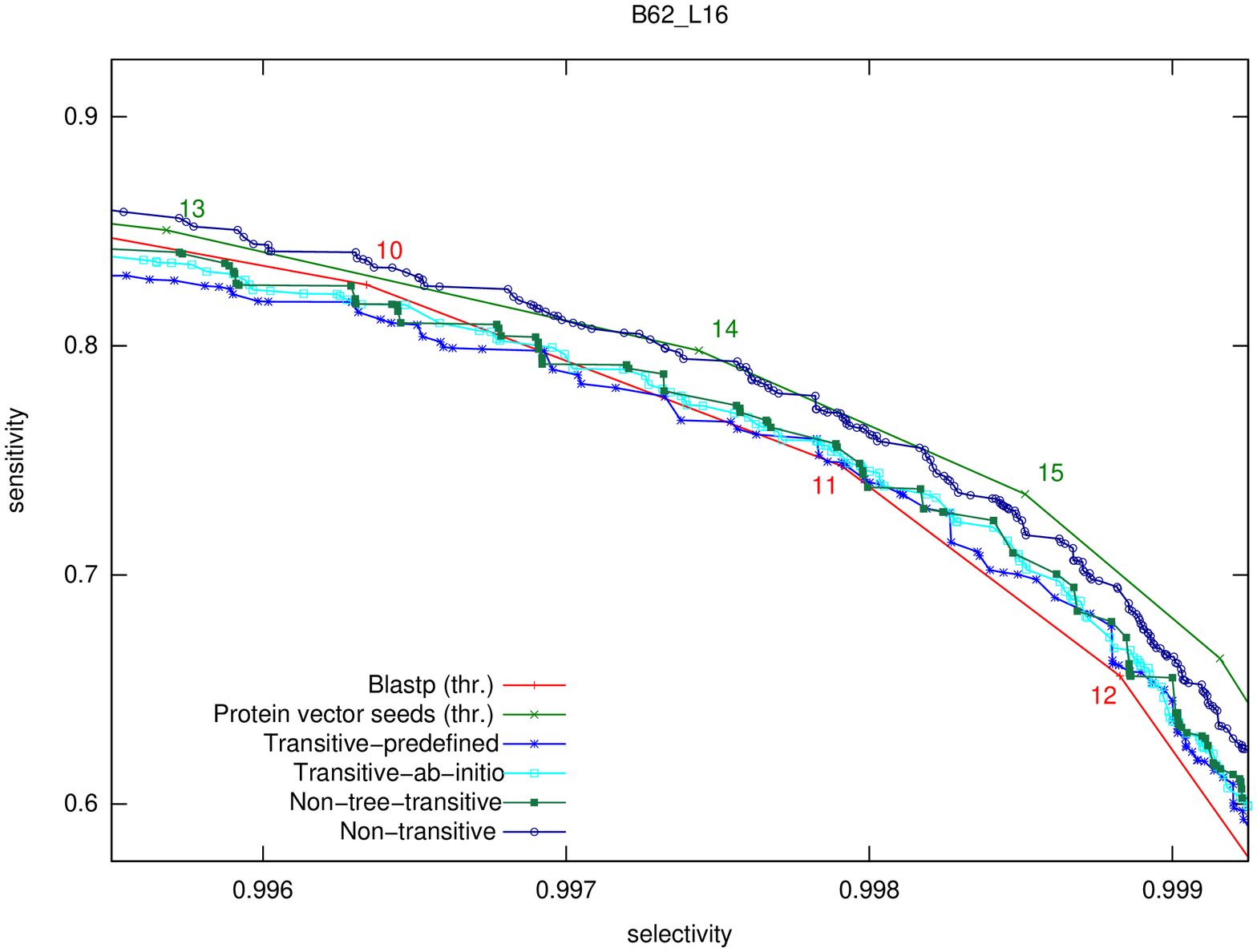}
\includegraphics[width=12cm]{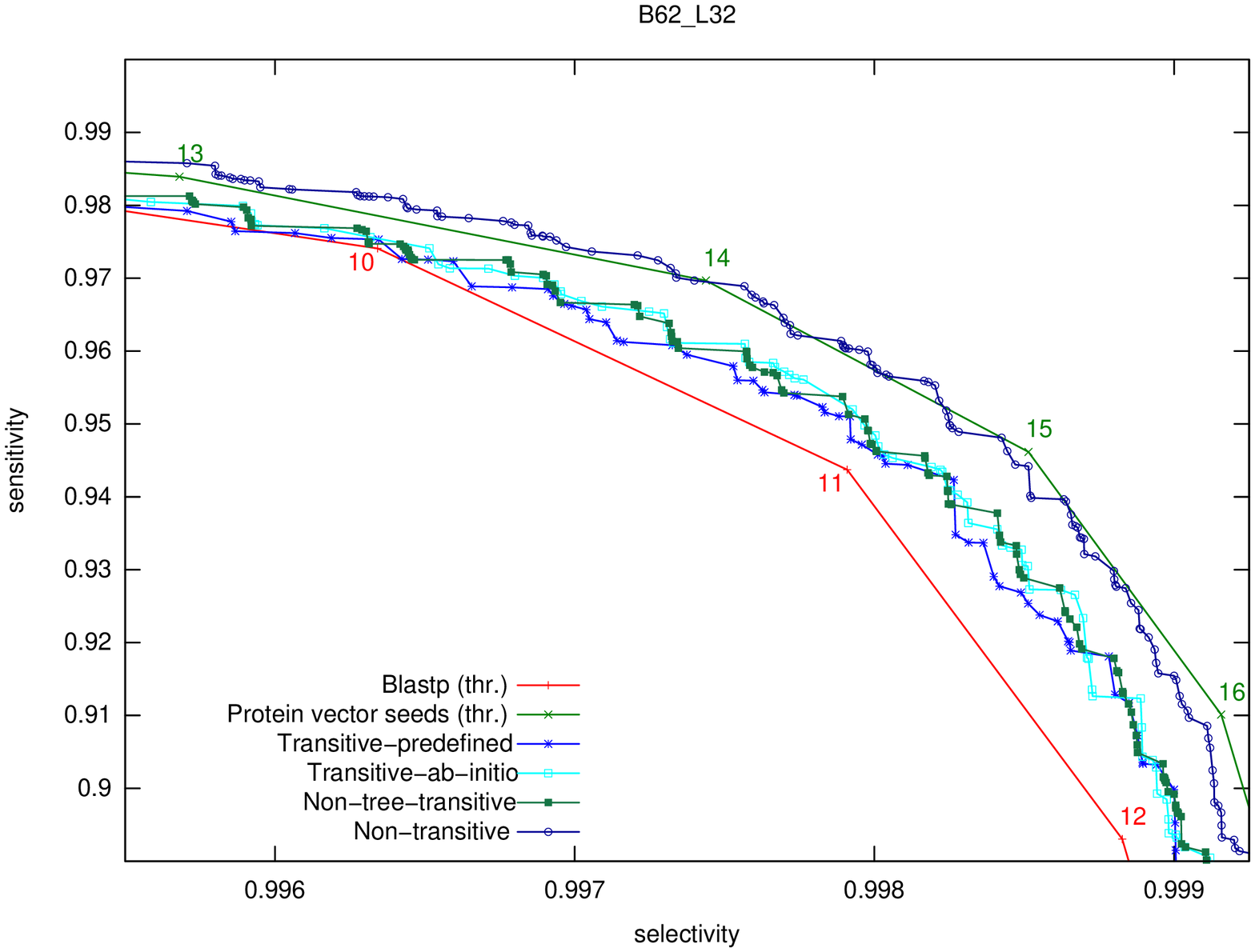}
\caption[ROC curves on theoretical
  models]{\label{figure:ROC-theoretical-zoom}ROC curves of seed
  performance measured on probabilistic models for alignment length 16 (above) and 32 (below). Blue, cyan dark green and dark bleu curves represent Pareto-optimal seed families constructed respectively over alphabets {\polish}, {\russianT}, {\russianNT} and {\nontrans}. Each point of these curves corresponds to a seed family, typically 3 to 5 seeds (respectively, 3 to 6 seeds) for alignment length 16 (respectively 32). Red and green polylines show the 
performance of {\sc Blastp} with word size 3 and the vector seed
family from \cite{BrownTCBB05}, for different score thresholds.}
\label{ROC-curves}
\end{figure}

We observe that seeds over the alphabet of
Section~\ref{section:non-transitive-seed-alphabet}  (dark blue curve)
are comparable in performance with the vector seed
family from \cite{BrownTCBB05} and clearly outperform seeds over
other alphabets. This result is interesting in itself, although in many cases this
alphabet is not practical due to its incompatibility with
the transitivity condition. 

As for the other alphabets, they roughly show a comparable
performance among them. 
For the alignment length 16, our seeds perform comparably to
{\sc Blastp}, with a slightly better performance for high thresholds
and a slightly worse performance for low thresholds. On the other
hand, for alignments of length 32, our seeds clearly outperform {\sc
  Blastp}. 
Note that the non-hierarchical alphabet from
Section~\ref{subsection:non-hierarchical} does not bring much of
improvement, which might indicate that lifting condition
(\ref{equation:consistent}) does not bring much of additional
power. This point, however, requires further investigation. 

\subsection{Results on real data}
\label{subsection:results_real_data}
We made large-scale tests of our seeds on real data by applying them
to several main databases of protein alignments. Those databases are
{\sc {B}ali{BASE}} (version 3)~\cite{BALIBASE01}, 
{\sc HOMSTRAD}~\cite{HOMSTRAD04}, 
{\sc {IRMB}ase} (version 1)~\cite{IRMBASE05},
{\sc {OXB}ench} (version 1.3)~\cite{OXBENCH03}, 
{\sc PFAM} (release 22)~\cite{PFAM06}, 
{\sc PREFAB} (version 4)~\cite{PREFAB04},
and {\sc SMART} (version 4)~\cite{SMART06}.

First, since all above databases except for {\sc OXBench} contain 
{\em multiple} 
alignments, we extracted from each of them a dataset of 
{\em pairwise} alignments. For this, pairs of aligned sequences have
been randomly extracted from multiple alignments and matching gaps
removed. To avoid a bias induced by big (in terms of the number of
sequences) multiple alignments, we selected a smaller fraction of
pairwise alignments from big multiple alignments than from small
ones: the number of selected alignments varied from order of $n^2$ for
small alignments to $\sqrt n$ for big ones. 
The total number of alignment
processed in our experiments varied from 640 ({\sc IRMBase}) to more
than 250000 ({\sc PFAM}).

For all those datasets, we identified alignments detected by the
{\sc Blastp} seed for different score thresholds (word length 3,
BLOSUM62 matrix, score threshold 10 to 13). 
On the other hand, for each {\sc Blastp} score threshold, 
we identified the closest seed family in the Pareto set (cf
Section~\ref{seed-design}) with equivalent or greater
selectivity. This has been done for each of the three transitive
alphabets proposed in
Section~\ref{section:transitive-seed-alphabet}. Selected seeds can be
found at the supplementary material Web page
{\tt\small\url{http://bioinfo.lifl.fr/yass/iedera_proteins/}}. 

Results are shown on Figure~\ref{figure:sensitivity-real-data}. 
Both methods detect a very high fraction of alignments of {\sc
  IRMBase} (all of them for thresholds 10 and 11). 
The poorest sensitivity 
is observed on {\sc SMART} where alignments represent small sequences of
proteins domains of the same family. A relatively weak 
sensitivity on {\sc PREFAB} is due to its method of obtaining
alignments which is based on structural information and, at the first
step, ``does not incorporate sequence similarity''. 
Finally, {\sc HOMSTRAD} combines both structural information (using
FUGUE~\cite{JOYFUGUE01}) and sequence information (using
  PSI-BLAST~\cite{GBLAST97}) which explains a better performance of
  seed-based search in this case.

\begin{figure}[b!]
  \includegraphics[width=14cm]{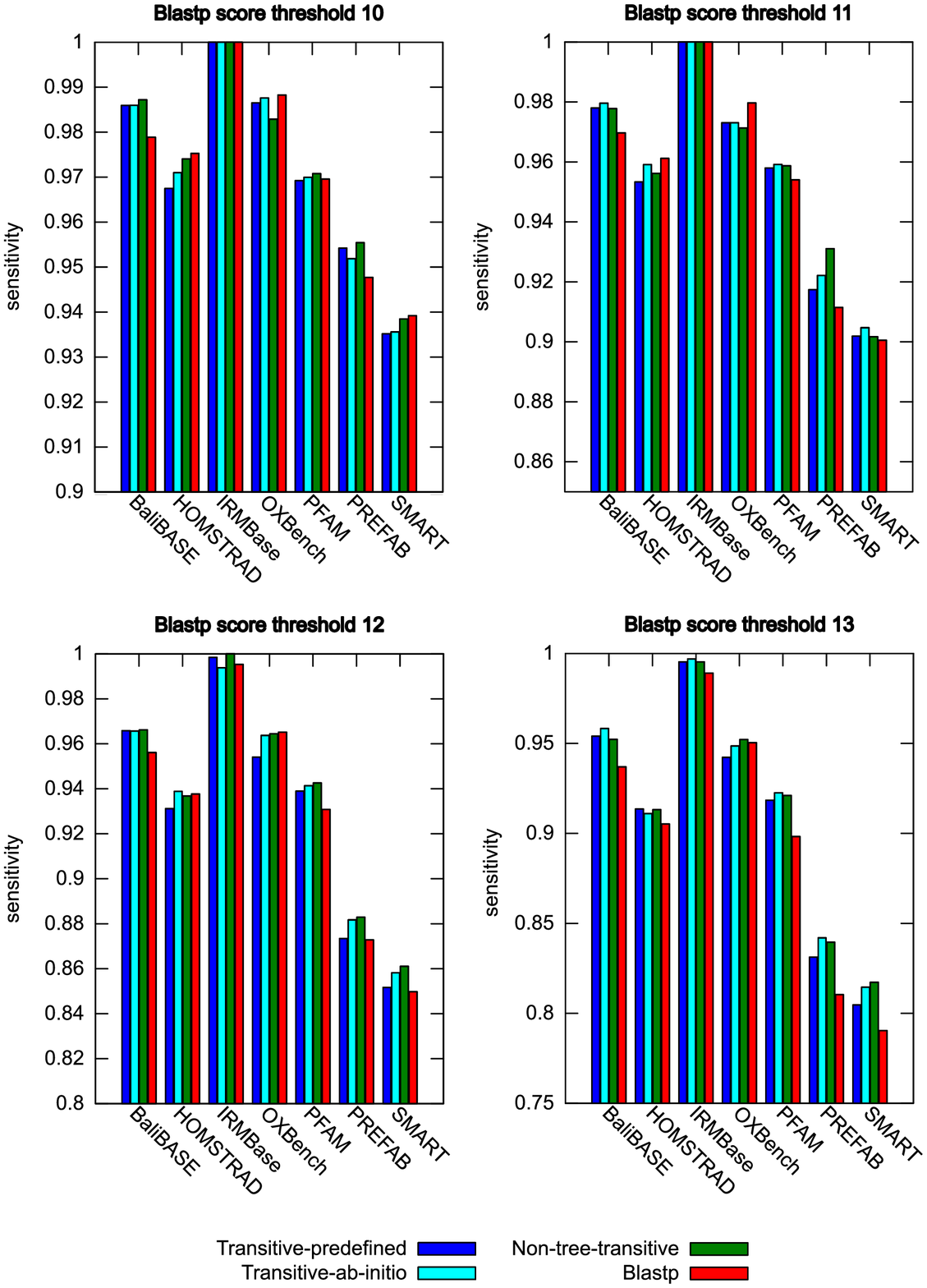}
  \caption{\label{figure:sensitivity-real-data}Sensitivity of subset seeds vs {\sc Blastp} measured on benchmark alignment databases}
\end{figure}
Comparing the performance of subset
seeds vs. {\sc Blastp}, the former show clearly a better performance on {\sc BaliBASE}, {\sc
  PREFAB} and {\sc PFAM}. On {\sc OXBench}, {\sc HOMSTRAD} and 
{\sc SMART}, the obtained sensitivity is very close
to that of {\sc Blastp}. Globally, subset seeds show a better
performance for higher selectivity levels (greater thresholds).

\section{Conclusion}

In this paper, we studied the design of {\em subset seeds} for
protein alignments, which is a very attractive seeding principle that
does not use scores at the hitting stage of the alignment
procedure. The design of efficient subset seeds subsumes a
design of appropriate {\em seed alphabets}, i.e. sets of {\em seed
  letters} that seeds can be built from. In this paper, we studied
several approaches to designing alphabets. 
In
Section~\ref{section:non-transitive-seed-alphabet}, we considered the most
general case when seed letters are only required to
induce a symmetric binary relation on amino acids. In
Section~\ref{section:transitive-seed-alphabet}, we focused on {\em
  transitive seed alphabets}, where seed letters are required to
induce an equivalence relation. In
Section~\ref{subsection:prebuilded-tree} we proposed an alphabet construction based on
{\em pre-defined} hierarchical clusterings of amino acids, while in
Section~\ref{subsection:ab-initio-clustering}, we considered a
construction based on an {\em ad hoc} clustering of amino acids based
on the likelihood ratio measure. Finally, in
Section~\ref{subsection:non-hierarchical} we lifted the requirement of
hierarchical clustering and considered alphabets with possibly
``incompatible'' letters (in the sense of embedding of equivalence
classes). 

The main conclusion of our work is that although the subset seed model
is less expressive than the method of cumulative score used in 
{\sc Blastp}, carefully designed subset seeds can reach the same or
even a higher performance. To put it informally, the use of the
cumulative score in defining a hit can, without loss of performance,
be replaced by a careful distinction between different amino acid
matches without using any scoring system. From a practical point of 
view, subset seeds can provide a more efficient implementation,
especially for large-scale protein comparisons, due to
a much smaller number of accesses to the hash table. In particular,
this can be very useful for parallel implementations or specialized
hardware (see e.g. \cite{PeterlongoEtAlPBC07,NguyenLavenierRIVF08}). 

Interestingly, the {\sc Blast} team reported recently in
\cite{ShiryevEtAlBioinfo07} that they used a reduced amino acid
alphabet in order to allow for longer seeds while still keeping the
hash table of acceptable size. (Note also that this idea has recently
been independently applied in \cite{PeterlongoEtAllBMCBioinformatics08}, in a
slightly different context.) This is done, however, by translating
one of the sequences into a compressed alphabet and still using
neighborhoods and a cumulative hit criterion. In this work, we
demonstrated that instead of that, one can apply
carefully designed subset seeds to avoid using neighborhoods and
scoring systems at the seeding stage, without sacrificing the
performance.  

Note that the seed design heuristic sketched in
Section~\ref{seed-design} does not guarantee to compute optimal
seeds, and therefore our seeds could potentially be further improved
by a more advanced design procedure, possibly bringing a further
increase in performance. This is especially true for seeds of large
weight (due to a bigger number of those), for which our seed design
procedure could produce non-optimal seeds, thus explaining some
``drop-offs'' in high-selectivity parts of plots of
Figure~\ref{ROC-curves}.  


As far as further research is concerned, the question of efficient
seed design remains an open issue. Improvements of the hill climbing
heuristic used in this work are likely to be possible. 



Finally, it would be very interesting to further study the
relationship between optimal seeds and seed letters those seeds
contain. In particular, it often appeared in our experiments that
optimal seeds contained ``non-optimal'' seed letters. Understanding
this phenomenon is an interesting theoretical question for further 
study. 

\textit{Acknowledgments.}
Parts of this work have been done during visits to LIFL of Ewa
Szczurek (June-August 2006), Anna Gambin and S\l{}awomir Lasota
(August 2006) and Mikhail Roytberg (October-December 2006). These
visits were supported by the ECO-NET and Polonium programs of the
French Ministry of Foreign Affairs.
Laurent No\'e was supported by the ANR project CoCoGen (BLAN07-1\_185484).
Mikhail Roytberg and Eugenia Furletova were supported by grants RFBR
06-04-49249 and 08-01-92496, and INTAS 05-10000008-8028.  The authors thank Ivan
Tsitovich for fruitful discussions of statistical questions related to
this work, and Mathieu Giraud and Marta Girdea for commenting on the 
manuscript. 

\bibliographystyle{IEEEtran}
\bibliography{paper}
\end{document}